\documentclass[envcountsect,envcountsame,oribibl]{llncs}

\usepackage{amsmath,amsfonts,amssymb}

\newenvironment{proofof}[1]%
      {\par\noindent{\it Proof of #1.}} %
      {\qed \medskip}

\newcommand{\Exp}{\mathbb{E}}
\newcommand{\e}{\mathrm{e}}
\newcommand{\phimax}{\phi_{\max}}
\newcommand{\zetamax}{\zeta_{\max}}

\newcommand{\eps}{\varepsilon}
\newcommand{\ccrit}{c_\mathrm{crit}}
\newcommand{\acrit}{\alpha_\mathrm{crit}}
\newcommand{\tG}{\widetilde{G}}
\newcommand{\abs}[1]{\left| #1 \right|}
\newcommand{\about}[1]{\mathrm{\Theta}\!\left( #1 \right)}

\begin{document}

\title{Independent Sets in Random Graphs from the Weighted Second Moment Method}
\author{Varsha Dani\inst{1} 
\and 
Cristopher Moore\inst{1, 2}}
\institute{Computer Science Department, University of New Mexico 
\and Santa Fe Institute}
\maketitle

\begin{abstract}
We prove new lower bounds on the likely size of the maximum independent set in a random graph with a given constant average degree.  Our method is a weighted version of the second moment method, where we give each independent set a weight based on the total degree of its vertices.
\end{abstract}

\section{Introduction}

We are interested in the likely size of the largest independent set $S$ in a random graph with a given average degree. It is easy to see that $|S|=\mathrm{\Theta}(n)$ whenever the average degree is constant, and Shamir and Spencer~\cite{shamirSpencer87} showed using Azuma's inequality that, for any fixed $n$, $|S|$ is tightly concentrated around its mean.  Moreover, Bayati, Gamarnik, and Tetali~\cite{bayati2010} recently showed that $|S|/n$ converges to a limit with high probability.  Thus for each constant $c$ there is a constant $\acrit = \acrit(c)$ such that    
\[
\lim_{n \rightarrow \infty} \Pr[\mbox{$G(n,p=c/n)$ has an independent set of size $\alpha n$}] 
= \begin{cases} 1 & \alpha < \acrit \\
0 & \alpha > \acrit \, .
\end{cases}
\]
By standard arguments this holds in $G(n,m=cn/2)$ as well.  

Our goal is to bound $\acrit$ as a function of $c$, or equivalently to bound
\[
\ccrit = \sup \,\{ c : \acrit(c) \ge \alpha \} \, ,
\]
as a function of $\alpha$.  For $c \le \e$, a greedy algorithm of Karp and Sipser~\cite{karp-sipser} asymptotically finds a maximal independent set, and analyzing this algorithm with differential equations yields the exact value of $\acrit$.  For larger $c$, Frieze~\cite{frieze90} determined $\acrit$ to within $o(1/c)$, where $o$ refers to the limit where $c$ is large.  These bounds were improved by Coja-Oghlan and Efthymiou~\cite{coja2011} who prove detailed results on the structure of the set of independent sets.

We further improve these bounds.  Our method is a weighted version of the second moment method, inspired by the work of Achlioptas and Peres~\cite{ach-peres} on random $k$-SAT, where each independent set is given a weight depending on the total degree of its vertices.  In addition to improving bounds on this particular problem, our hope is that this advances the art and science of inventing random variables that counteract local sources of correlation in random structures.

We work in a modified version of the $G(n,m)$ model which we call $\tG(n,m)$.  For each of the $m$ edges, we choose two vertices $u,v$ uniformly and independently and connect them.  This may lead to a few multiple edges or self-loops. A vertex with a self-loop cannot belong to an independent set.  In the sparse case where $m=cn/2$ for constant $c$, with constant positive probability $\tG(n,m=cn/2)$ has no multiple edges or self-loops, in which case it is uniform in the usual model $G(n,m=cn/2)$ where edges are chosen without replacement from distinct pairs of vertices.  Thus any property which holds with high probability for $\tG(n,m)$ also holds with high probability for $G(n,m)$, and any bounds we prove on $\acrit$ in $\tG(n,m)$ also hold in $G(n,m)$.

We review the first moment upper bound on $\acrit$ from Bollob\'as~\cite{bollobas}.  Let $X$ denote the number of independent sets of size $\alpha n$ in $\tG(n,m)$.  Then
\[
\Pr[X > 0] \le \Exp[X] \, .
\]
By linearity of expectation, $\Exp[X]$ is the sum over all $\binom{n}{\alpha n}$ sets of $\alpha n$ vertices of the probability that a given one is independent.  The $m$ edges $(u,v)$ are chosen independently and for each one $u,v \in S$ with probability $\alpha^2$, so
\[
\Exp[X] = \binom{n}{\alpha n} (1-\alpha^2)^m \, . 
\]
In the limit $n \to \infty$, Stirling's approximation $n! = (1+o(1)) \sqrt{2\pi n} \,n^n \,\e^{-n}$ gives
\begin{equation}
\label{eq:stirling}
\binom{n}{\alpha n} = \about{\frac{1}{\sqrt{n}} \,\e^{n h(\alpha)}} \, ,
\end{equation}
where $h$ is the entropy function
\[
h(\alpha) = -\alpha \ln \alpha - (1-\alpha) \ln (1-\alpha) \, ,
\]
and where $\mathrm{\Theta}$ hides constants that depend smoothly on $\alpha$.  Thus
\[
\Exp[X] = \about{ \frac{1}{\sqrt{n}} \,\e^{n(h(\alpha)+ (c/2) \ln (1-\alpha^2))} } \, .
\]
For each $c$, the $\alpha$ such that 
\begin{equation}
\label{eq:first-moment}
h(\alpha)+ (c/2) \ln (1-\alpha^2) = 0 
\end{equation}
is an upper bound on $\acrit(c)$, since for larger $\alpha$ the expectation $\Exp[X]$ is exponentially small.

We find it more convenient to parametrize our bounds in terms of the function $\ccrit(\alpha)$.  Then~\eqref{eq:first-moment} gives the following upper bound,
\begin{equation}
\label{eq:c-upper}
\ccrit(\alpha) 
\le 2 \,\frac{\alpha \ln \alpha + (1-\alpha) \ln (1-\alpha)}{\ln (1-\alpha^2)} 
\le 2 \,\frac{\ln (1/\alpha) + 1}{\alpha} \, .
\end{equation}
We will prove the following nearly-matching lower bound.  
\begin{theorem}
\label{thm:lower}
For any constant $x > 4/\e$, for sufficiently small $\alpha$
\begin{equation}
\label{eq:c-lower}
\ccrit(\alpha)
\ge 2 \,\frac{\ln (1/\alpha) + 1}{\alpha} - \frac{x}{\sqrt{\alpha}} \, .
\end{equation}
\end{theorem}
\noindent
Coja-Oghlan and Efthymiou~\cite{coja2011} bounded $\ccrit$ within a slightly larger factor $O(\sqrt{\ln (1/\alpha)/\alpha})$.
 
Inverting~\eqref{eq:c-upper} and Theorem~\ref{thm:lower} gives the following bounds on $\acrit(c)$.  The lower bound is a significant improvement over previous results:
\begin{corollary}
\label{cor:lower}
For $z > 0$, let $W(z)$ denote the unique positive root $x$ of the equation $x \e^x = z$. Then for any constant $y > 4 \sqrt{2} / \e$, 
\[
%\abs{ \acrit(c) - \frac{2}{c} \left( \ln c - \ln \ln c + 1 - \ln 2 + \frac{\ln \ln c}{\ln c} - \frac{1-\ln 2}{\ln c} \right) } < \frac{\eps}{c \ln c} \, . 
%\acrit \ge \frac{2}{c} \left( \ln c - \ln \ln c + 1 - \ln 2 + \frac{\kappa(c)}{\ln c} - y \sqrt{\frac{\ln c}{c}} \right) \, ,
\frac{2}{c} \,W\!\left( \frac{\e c}{2} \right) - y \,\frac{\sqrt{\ln c}}{c^{3/2}} 
\le \acrit 
\le \frac{2}{c} \,W\!\left( \frac{\e c}{2} \right) \, ,
\]
where the lower bound holds for sufficiently large $c$.
%where
%\begin{align*}
%\kappa(c) 
%%= (\ln \alpha-\ln (2-2 \ln \alpha)) (\ln (1-\ln \alpha) -\ln (\ln (2-2 \ln \alpha)-\ln \alpha))
%&= \ln\!\left[ 2 \,\frac{\ln (1/\alpha) + 1}{\alpha} \right] 
%\; \ln\!\left[  \left. \ln\!\left[ 2 \,\frac{\ln (1/\alpha) + 1}{\alpha} \right] \right\slash (\ln (1/\alpha) + 1) \right] \\
%&= \ln \ln (1/\alpha) - 1 + \ln 2 
%\]
\end{corollary}
\noindent
If we like we can expand $W(\e c/2)$ asymptotically in $c$, 
\begin{align*}
W\!\left( \frac{\e c}{2} \right) 
&= \ln c - \ln \ln c + 1 - \ln 2 + \frac{\ln \ln c}{\ln c} - \frac{1-\ln 2}{\ln c} \\
&+ \frac{1}{2} \frac{(\ln \ln c)^2}{(\ln c)^2} - (2-\ln 2) \frac{\ln \ln c}{(\ln c)^2} 
+ \frac{3 + (\ln 2)^2 - 4 \ln 2}{2 (\ln c)^2} + O\!\left( \frac{(\ln \ln c)^3}{(\ln c)^3} \right) 
\, . 
\end{align*}
The first few of these terms correspond to the bound in~\cite{frieze90}, and we can extract as many additional terms as we wish.

%The aesthetics of expressing  this way may be questionable.  
%In fact, Theorem~\ref{thm:lower} implies much tighter bounds, determining $\acrit(c)$ within 
%\[
%O\!\left(\frac{1}{\sqrt{\alpha}} \frac{\partial \alpha}{\partial c} \right) 
%= O(\sqrt{\ln c} / c^{3/2}) \, , 
%\]
%but extracting additional terms becomes a tedious exercise in algebra.  

\section{The weighted second moment method}

Our proof uses the second moment method.  For any nonnegative random variable $X$, the Cauchy-Schwarz inequality implies that 
\begin{equation}
\label{eq:prob-lower-bound}
\Pr[X > 0] \ge \frac{\Exp[X]^2}{\Exp[X^2]} \, .
\end{equation}
If $X$ counts the number of objects of a certain kind, \eqref{eq:prob-lower-bound} shows that at least one such object exists as long as the expected number of objects is large and the variance is not too large.

Unfortunately, applying this directly to the number $X$ of independent sets fails utterly.  The problem is that for most pairs of sets of size $\alpha n$, the events that they are independent are highly correlated, unlike the case where the average degree grows sufficiently quickly with $n$~\cite{bollobas,frieze90}.  As a result, $\Exp[X^2]$ is exponentially larger than $\Exp[X]^2$, and the second moment method yields an exponentially small lower bound on $\Pr[X > 0]$.

One way to deal with these correlations, used by Frieze in~\cite{frieze90}, is to partition the vertices into sets $V_i$ of $\lfloor 1/\alpha \rfloor$ vertices each and focus on those independent sets that intersect each $V_i$ exactly once.  In that case, a large-deviations inequality allows us to override the correlations.  

Here we pursue a different approach, inspired by the work of Achlioptas and Peres~\cite{ach-peres} on random $k$-SAT.  The idea is to give each independent set $S$ 
%of size $\alpha n$ 
a weight $w(S)$, depending exponentially on local quantities in the graph.  Specifically, we define
\[
w(S) = \mu^{\mbox{\scriptsize \# of edges $(u,v)$ with $u,v \notin S$}} \, ,
\]
for some $\mu < 1$.  If the number of edges $m$ is fixed, the number of edges where neither endpoint is in $S$ is simply $m$ minus the total degree of the vertices in $S$.  Thus we can also write 
\begin{equation}
\label{eq:weight}
w(S) = \mu^{m-\sum_{v \in S} \deg(v)} \, .
\end{equation}

We will apply the second moment method to the total weight of all independent sets of size $\alpha n$,
\[
X = \sum_{\substack{S \subseteq V, |S| = \alpha n \\ S\, \textrm{independent}}} w(S) \, .
\]
%\[
%X = \sum_{\substack{S \subseteq V \\ S\, \textrm{independent}}} w(S) \, .
%\]
If we tune $\mu$ properly, then for particular $\alpha^\star, c^\star$ we have $\Exp[X^2] = \about{\Exp[X]^2}$, in which case $\Pr[X > 0]$ is bounded above zero.  
%Since the fractional size $\alpha$ of the largest independent set is tightly concentrated, it follows that an independent set of size $\alpha^\star n$ exists in $\tG(n,m=c^\star n/2)$ with high probability.  
In that case $\ccrit(\alpha^\star) \ge c^\star$, or equivalently $\acrit(c^\star) \ge \alpha^\star$.

Why is this the right type of weight?  Intuitively, one of the main sources of correlations between independent sets is the temptation to occupy low-degree vertices.  For instance, any two maximal independent sets contain all the degree-zero vertices, giving them a large overlap.  If $X$ simply counts the independent sets of size $\alpha n$, the resulting correlations make $X$'s variance exponentially large compared to the square of its expectation, and the second moment fails.

Weighting each $S$ as in~\eqref{eq:weight} counteracts this temptation, punishing sets that occupy low-degree vertices by reducing their weight exponentially.  As we will see below, when $\mu$ is tuned to a particular value, making this punishment condign, these correlations disappear in the sense that the dominant contribution to $\Exp[X^2]$ comes from pairs of sets $S, T$ of size $\alpha n$ such that $\abs{S \cap T} = \alpha^2 n + O(\sqrt{n})$, just as if $S$ and $T$ were chosen independently from among all sets of size $\alpha n$.  

This is analogous to the situation for $k$-SAT, where satisfying assignments are correlated because of the temptation to give each variable the truth value that agrees with the majority of its literals in the formula.  By giving each satisfying assignment a weight $\eta^{\mbox{\scriptsize \# of true literals}}$ and tuning $\eta$ properly, we make the dominant contribution to $\Exp[X^2]$ come from pairs of satisfying assignments which agree on $n/2+O(\sqrt{n})$ variables, just as if they were chosen independently~\cite{ach-peres}.

Proceeding, let us compute the first and second moments of our random variable $X$.  We extend the weight function $w(S)$ to all sets $S \subseteq V$ by setting $w(S)=0$ if $S$ is not independent.  That is, 
\[
X = \sum_{\substack{S \subseteq V \\ |S| = \alpha n}} w(S) 
\]
where
\[
w(S) = \prod_{(u,v) \in E} w_{u,v}(S) 
\]
and
\begin{equation}
\label{eq:wuv}
w_{u,v}(S)  = \begin{cases} 
\mu & \mbox{if $u, v \notin S$} \\ 
1 & \mbox{if $u \in S, v \notin S$ or vice versa} \\
0 & \mbox{if $u,v \in S$} \, .
\end{cases}
\end{equation}

We start by computing $\Exp[X]$.  Fix a set $S$ of size $\alpha n$.  Since the $m$ edges are chosen independently, 
\[
\Exp[w(S)] = w_1(\alpha,\mu)^m 
\quad
\text{where}
\quad
w_1(\alpha,\mu) = \Exp_{u,v}[w_{u,v}(S)] \, .
\]
For each edge $(u,v)$ in $\tG(n,m)$, $u$ and $v$ are chosen randomly and independently, so the probabilities of the three cases in~\eqref{eq:wuv} are $(1-\alpha)^2$, $2\alpha(1-\alpha)$, and $\alpha^2$ respectively.  Thus 
\[
w_1(\alpha, \mu)
%\Exp_{u,v}[w_{u,v}(S)] 
= (1-\alpha)^2 \mu + 2\alpha(1-\alpha) \, . 
\]
%Since the edges in $\tG(n,m)$ are independent, 
%\[
%\Exp[w(S)] = \left( \Exp_{u,v}[w_{u,v}(S)] \right)^m 
%%=( (1-\alpha)^2 \mu + 2\alpha(1-\alpha) )^m \equiv (w_1(\alpha, \mu))^m
%= w_1(\alpha, \mu)^m \, .
%\]
By linearity of expectation,
\[
\Exp[X] = \sum_{\substack{S \subseteq V \\ |S| = \alpha n}} \Exp[w(S)] = \binom{n}{\alpha n} \,w_1(\alpha, \mu)^m \,.
\]
Using Stirling's approximation~\eqref{eq:stirling} 
%, $\binom{n}{\alpha n} = \about{ \frac{1}{\sqrt{n}} \,\e^{n h(\alpha)} }$ 
and substituting $m = cn/2$ gives
\begin{equation}
\label{eq:x-first}
\Exp[X] = \about{ \frac{1}{\sqrt{n}} \,\e^{n f_1(\alpha)} }
\quad \text{where} \quad
f_1(\alpha) = h(\alpha)+ \frac{c}{2} \,\ln w_1(\alpha, \mu) \, .
\end{equation}
As before, $\mathrm{\Theta}$ hides constant factors that depend smoothly on $\alpha$.

Next we compute the second moment.  We have 
\[
\Exp[X^2] = \Exp\left[ \sum_{S} w(S) \sum_{T}w(T)\right] = \sum_{S,T} \Exp[w(S) \,w(T)]
\]
where $S$ and $T$ are subsets of $V$ of size $\alpha n$. The expectation of $w(S) \,w(T)$ does not depend on the specific choice of $S$ and $T$, but it does depend on the size of their intersection.  We say that $S$ and $T$ have \emph{overlap} $\zeta$ if $\abs{S \cap T} = \zeta n$.  Again using the independence of the edges, we have
\[
\Exp[w(S) \,w(T)] = w_2(\alpha, \zeta, \mu)^m 
\quad
\text{where}
\quad
w_2(\alpha, \zeta, \mu) = \Exp_{u,v}\left[ w_{u,v}(S) \,w_{u,v}(T) \right] \, . 
\]

For each edge $(u,v)$ of $\tG$, the probability that it has no endpoints in $S$ or $T$ is $(1-2\alpha +\zeta)^2$, in which case it contributes $\mu^2$ to $w_{u,v}(S) \,w_{u,v}(T)$.  The probability that it has one endpoint in $S$ and none in $T$ or vice versa is $2(2\alpha - 2\zeta)(1- 2\alpha +\zeta)$, in which case it contributes $\mu$. Finally, the probability that it has one endpoint in $S$ and one in $T$ is $2(\alpha - \zeta)^2 + 2\zeta(1- 2\alpha +\zeta)$, in which case it contributes $1$. With the remaining probability it has both endpoints in $S$ or $T$, in causing them to be non-independent and contributing zero.  Thus
\[
w_2(\alpha, \zeta, \mu) 
%&\equiv \Exp[w_e(S)w_e(T)] \\ &
= (1-2\alpha +\zeta)^2 \mu^2 + 4(\alpha - \zeta)(1- 2\alpha +\zeta)\mu + 2(\alpha - \zeta)^2 + 2\zeta(1- 2\alpha +\zeta) 
\]
Observe that when $\zeta = \alpha^2$, as it typically would be if $S$ and $T$ were chosen independently and uniformly, we have 
\begin{equation}
\label{eq:w2-w1}
w_2 = w_1^2 \, . 
\end{equation}

The number of pairs of sets $S, T$ of size $\alpha n$ and intersection of size $z = \zeta n $ is the multinomial
\[
\binom{n}{\zeta n, (\alpha -\zeta) n, (\alpha -\zeta)n, (1-2\alpha +\zeta)n}
= \binom{n}{\alpha n} \binom{\alpha n}{\zeta n} \binom{(1-\alpha)n}{(\alpha-\zeta) n} \, ,
%&\sim \frac{f(\zeta)}{\sqrt{n^3}}e^{n \left(h(\alpha) + \alpha h\left(\frac{\zeta}{\alpha}\right) + (1-\alpha)h\left(\frac{\alpha -\zeta}{1-\alpha}\right) \right)} \,.
\]
and linearity of expectation gives
\[
\Exp[X^2] 
= \sum_{z=0}^{\alpha n} 
\binom{n}{z, \alpha n-z, \alpha n - z, (1-2\alpha)n + z} 
\,w_2(\alpha, \zeta, \mu)^m \, . 
\]
This sum is dominated by the terms where $\zeta = z/n$ is bounded inside the interval $(0,\alpha)$.  Stirling's approximation then gives
\[
%\label{eq:multi-stirling}
\binom{n}{\zeta n, (\alpha -\zeta) n, (\alpha -\zeta)n , (1-2\alpha +\zeta)n } 
= \about{ \frac{\e^{n \left[ h(\alpha) + \alpha h\left(\zeta/\alpha\right) + (1-\alpha) h\left(\frac{\alpha -\zeta}{1-\alpha}\right) \right]}}{n^{3/2}} }
%\,\e^{n \left[ h(\alpha) + \alpha h\left(\zeta/\alpha\right) + (1-\alpha) h\left(\frac{\alpha -\zeta}{1-\alpha}\right) \right]} }
\, ,
\]
where $\mathrm{\Theta}$ hides constants that vary slowly with $\alpha$ and $\zeta$.  
%The last expression comes from using Stirling's approximation 
%$n! \sim \sqrt{2\pi n} \,\e^{-n}n^n$, and 
%\[
%f(\zeta) = 
%\left(8 \pi^3 \zeta(\alpha - \zeta)^2(1-2\alpha +\zeta)\right)^{-1/2} \,.
%\]
Thus the contribution to $\Exp[X^2]$ of pairs of sets with overlap $\zeta \in (0,\alpha)$ is 
\begin{equation}
\label{eq:x-second}
\frac{1}{n^{3/2}}
\,\e^{n f_2(\alpha,\zeta,\mu)}
\end{equation}
where 
\[
f_2(\alpha,\zeta,\mu) 
= h(\alpha) + \alpha h\!\left(\frac{\zeta}{\alpha}\right) + (1-\alpha) h\!\left(\frac{\alpha -\zeta}{1-\alpha}\right) 
+ \frac{c}{2} \ln w_2(\alpha, \zeta, \mu) \, . 
\]
Combining~\eqref{eq:x-second} with~\eqref{eq:x-first}, we can write
\begin{equation}
\label{eq:moments-ratio}
\frac{\Exp[X^2]}{\Exp[X]^2} 
= \about{ \frac{1}{\sqrt{n}} \sum_{z=0}^{\alpha n} \e^{n \phi(z/n)} } \, ,
\end{equation}
where 
\begin{align*}
\phi(\zeta)
&= f_2(\alpha,\zeta,\mu) - 2 f_1(\alpha,\mu)  \\
&= \alpha h\!\left(\frac{\zeta}{\alpha}\right) + (1-\alpha) \,h\!\left(\frac{\alpha -\zeta}{1-\alpha}\right) - h(\alpha) 
+ \frac{c}{2} \ln \frac{w_2(\alpha, \zeta, \mu)}{w_1(\alpha, \mu)^2} \, .
%= h(\alpha) + \alpha h\left(\frac{\zeta}{\alpha}\right) + (1-\alpha)h\left(\frac{\alpha -\zeta}{1-\alpha}\right) 
%+ \frac{c}{2} \ln w_2(\alpha, \zeta, \mu) - 2 h(\alpha) -c \ln w_1(\alpha, \mu) \,.
\end{align*}
Using~\eqref{eq:w2-w1} and the fact that the entropy terms cancel, we have
\[
\phi(\alpha^2) = 0 \, .
\]
In other words, the contribution to $\Exp[X^2]$ from pairs of sets with overlap $\alpha^2$ is proportional to $\Exp[X]^2$.

We can now replace the sum in~\eqref{eq:moments-ratio} with an integral, 
\[
\frac{\Exp[X^2]}{\Exp[X]^2} 
= \about{ \frac{1}{\sqrt{n}} \sum_{z=0}^{\alpha n} \e^{n \phi(z/n)} }
= \about{ \sqrt{n} \int_0^\alpha \e^{n \phi(\zeta)} \mathrm{d}\zeta } \, , 
\]
and evaluate this integral using Laplace's method as in~\cite[Lemma 3]{ach-moore}.   Its asymptotic behavior depends on the maximum value of $\phi$, 
\[
\phimax = \max_{\zeta \in [0,\alpha]} \phi(\zeta) \, .  
\]
If $\phi'' < 0$ at the corresponding $\zetamax$, then it is dominated by an interval of width $\about{1/\sqrt{n}}$ around $\zetamax$ and 
\[
\frac{\Exp[X^2]}{\Exp[X]^2} 
%\sim \sqrt{n} \int_0^\alpha \e^{n \phi(\zeta)} \mathrm{d}\zeta 
= \about{ \e^{n \phimax} } \, . 
\]
If $\phimax = \phi(\alpha^2) = 0$, then $\Exp[X^2] = \about{ \Exp[X]^2 }$ and the second moment method succeeds.  Thus our goal is to show that $\phi$ is maximized at $\alpha^2$.

For this to happen, we at least need $\zeta = \alpha^2$ to be a \emph{local} maximum of $\phi$.  In particular, we need
\begin{equation}
\label{eq:local-max}
\phi'(\alpha^2) = 0 \, .
\end{equation}
%and $\phi''(\alpha^2) < 0$.  
Differentiating, we find that~\eqref{eq:local-max} holds if
\[
\mu = \frac{1-2\alpha}{1-\alpha} \, . 
\]
Henceforth, we will fix $\mu$ to this value.  In that case we have
\[
w_1 = 1-\alpha
\quad \text{and} \quad
w_2 = (1-\alpha)^2 + \frac{(\zeta - \alpha^2)^2}{(1-\alpha)^2} \, , 
\] 
so
\[
\phi(\zeta) 
= \alpha h\!\left(\frac{\zeta}{\alpha}\right) + (1-\alpha) \,h\!\left(\frac{\alpha -\zeta}{1-\alpha}\right) - h(\alpha) 
+ \frac{c}{2} \,\ln \!\left(1 +\frac{(\zeta -\alpha^2)^2}{(1-\alpha)^4} \right)
\] 
The remainder of this paper is dedicated to showing that for sufficiently small $\alpha$ as a function of $c$ or vice versa, $\phi$ is indeed maximized at $\alpha^2$.

\section{Finding and bounding the maxima}

Using $\ln (1+x) \le x$, we write $\phi(\zeta) \le \psi(\zeta)$ where
\begin{equation}
\label{eq:psi}
\psi(\zeta) 
= \alpha h\!\left( \frac{\zeta}{\alpha} \right) + (1-\alpha) \,h\!\left( \frac{\alpha-\zeta}{1-\alpha} \right) - h(\alpha)
+ \frac{c}{2} \frac{(\zeta-\alpha^2)^2}{(1-\alpha)^4} 
\, . 
\end{equation}
Note that 
\[
\psi(\alpha^2) = \phi(\alpha^2) = 0 \, .
\]
Our goal is to show for an appropriate $c$ that $\zeta = \alpha^2$ is in fact the global maximum of $\psi$, and therefore of $\phi$.  
In what follows, asymptotic symbols such as $O$ and $o$ refer to the limit $\alpha \to 0$, or equivalently the limit $c \to \infty$.  Error terms may be positive or negative unless otherwise stated.

The first two derivatives of $\psi(\zeta)$ are
\begin{gather}
\psi'(\zeta) = 
\frac{c \left(\zeta-\alpha^2\right)}{(1-\alpha)^4} 
+ 2 \ln (\alpha-\zeta) - \ln \zeta - \ln (1-2 \alpha+\zeta) 
\label{eq:psi1}
\\
\psi''(\zeta) = \frac{c}{(1-\alpha)^4} - \frac{2}{\alpha-\zeta} - \frac{1}{\zeta} - \frac{1}{1-2 \alpha+\zeta} 
\label{eq:psi2}
%\psi'''(\zeta) = \frac{1}{(1-2 \alpha+\zeta)^2}-\frac{2}{(\alpha-\zeta)^2}+\frac{1}{\zeta^2}
\end{gather}
The second derivative $\psi''(\zeta)$ tends to $-\infty$ at $\zeta=0$ and $\zeta=\alpha$.  Setting $\psi''(\zeta)=0$ yields a cubic equation in $\zeta$ which has one negative root and, for sufficiently small $\alpha$, two positive roots in the interval $[0,\alpha]$.  Thus for each $\alpha$ and sufficiently small $\alpha$, there are $0 < \zeta_1 < \zeta_2 < \alpha$ where 
\[
\psi''(\zeta) \begin{cases}
< 0 & 0 < \zeta < \zeta_1 \\
> 0 & \zeta_1 < \zeta < \zeta_2 \\
< 0 & \zeta_2 < \zeta < \alpha \, .
\end{cases}
\]

It follows that $\psi$ can have at most two local maxima.  One is in the interval $[0,\zeta_1]$, and the following lemma shows that for the relevant $\alpha$ and $c$ this is $\alpha^2$:

\begin{lemma}
\label{lem:local}
If $c=o(1/\alpha^2)$ then for sufficiently small $\alpha$, $\zeta_1 > \alpha^2$ and $\psi(\alpha^2)$ is a local maximum.
\end{lemma}

\noindent 
The other local maximum is in the interval $[\zeta_2,\alpha]$, and we denote it $\zeta_3$.  To locate it, first we bound $\zeta_2$:

\begin{lemma}
\label{lem:zeta2}
If
\[
c = (2+o(1)) \,\frac{\ln (1/\alpha)}{\alpha} \, ,
\]
then 
\[
\frac{\zeta_2}{\alpha} = 1 - \delta_2 
\quad \text{where} \quad
\delta_2 = \frac{1+o(1)}{\ln (1/\alpha)} \, . 
\]
\end{lemma}

\noindent
Thus $\zeta_2/\alpha$, and therefore $\zeta_3/\alpha$, tends toward $1$ as $\alpha \to 0$.  

We can now locate $\zeta_3$ when $\alpha$ is close to its critical value.
\begin{lemma}
\label{lem:zeta3}
If
\[
c = \frac{1}{\alpha} \big( 2 \ln (1/\alpha) + 2-o(1) \big) \, ,
%\alpha = \frac{2}{c} \big( \ln c - \ln \ln c + 1 - \ln 2 + o(1) \big) \, ,
\]
then 
\[
\frac{\zeta_3}{\alpha} = 1 - \delta_3 
\quad \text{where} \quad
\delta_3 =  \frac{1+o(1)}{\e} \sqrt{\alpha} \, . 
\]
\end{lemma}

\begin{lemma}
\label{lem:done}
For any constant $x > 4/\e$, if 
\[
c = \frac{2 \ln (1/\alpha) + 2 - x \sqrt{\alpha} }{\alpha} \, ,
\]
then $\psi(\zeta_3) < 0$ for sufficiently small $\alpha$.
\end{lemma}

\section{Proofs}

\begin{proofof}{Lemma~\ref{lem:local}}
Setting $\zeta = \alpha^2$ in~\eqref{eq:psi2} gives
\[
\psi''(\alpha^2) 
< \frac{c}{(1-\alpha)^4} - \frac{1}{\alpha^2} \, .
\]
If $c = o(1/\alpha^2)$ this is negative for sufficiently small $\alpha$, in which case $\zeta_1 > \alpha^2$ and $\psi(\alpha^2)$ is a local maximum.
%$\alpha < 0.15$. 
\end{proofof}

\begin{proofof}{Lemma~\ref{lem:zeta2}}
For any constant $b$, if 
\begin{equation}
\label{eq:zeta2a}
\frac{\zeta}{\alpha} = 1 - \delta
\quad \text{where} \quad
\delta = \frac{b}{\ln (1/\alpha)} 
\end{equation}
then~\eqref{eq:psi2} gives
\[
\psi''(\zeta) 
= \left( 2-\frac{2}{b}+o(1)\right) \frac{\ln (1/\alpha)}{\alpha} - O(1/\alpha) \, .
\]
If $b \ne 1$, for sufficiently small $\alpha$ this is negative if $b < 1$ and positive if $b > 1$.  Therefore $\zeta_2 / \alpha = 1-\delta_3$ where $\delta_3 = (1+o(1))/\ln(1/\alpha)$.
\end{proofof}

\begin{proofof}{Lemma~\ref{lem:zeta3}}
Lemma~\ref{lem:zeta2} tells us that $\zeta_3 = \alpha(1-\delta)$ for some 
\[ 
\delta < \frac{1+o(1)}{\ln(1/\alpha)} \, . 
\]
Setting $\zeta = \alpha(1-\delta)$ in~\eqref{eq:psi1} and using 
\[
\frac{1}{(1-\alpha)^4} = 1 + O(\alpha) % + O(\alpha^2) 
\quad \text{and} \quad 
-\ln (1-x) = O(x) % + O(x^2) 
\]
gives, after some algebra, 
\[
\psi'(\zeta)
= \alpha c + \ln \alpha + 2 \ln \delta + O(\alpha \delta c) + O(\alpha^2 c) \, .
\]
For any constant $b$, setting
\[
\delta = \frac{b \sqrt{\alpha}}{\e} 
\]
gives
\[
\psi'(\zeta)
= \alpha c + 2 \ln \alpha + 2 \ln b - 2 + O(\alpha^{3/2} c)  \, ,
\]
and setting
\[
c = \frac{2 \ln (1/\alpha) + 2-\eps}{\alpha} 
%\alpha = \frac{2}{c} \big( \ln c - \ln \ln c + 1 - \ln 2 + o(1) \big) 
\]
then gives
\[
\psi'(\zeta) = 2 \ln b - \eps + o(1) \, .
\]
If $\eps = o(1)$ and $b \ne 1$, for sufficiently small $\alpha$ this is negative if $b < 1$ and positive if $b > 1$.  Therefore $\zeta_3 / \alpha = 1-\delta_3$ where $\delta_3 = (1+o(1)) \sqrt{a} / \e$.
\end{proofof}

\begin{proofof}{Lemma~\ref{lem:done}}
Setting $\zeta = \alpha(1-\delta)$ where $\delta = b \sqrt{a} / \e$ in~\eqref{eq:psi} and using the Taylor series 
\[
\frac{1}{(1-\alpha)^4} = 1 + 4 \alpha + O(\alpha^2) 
\quad \text{and} \quad
-\ln (1-x)=x+x^2/2+O(x^3)
\]
gives, after a fair amount of algebra,
\begin{align*}
\psi(\zeta) 
&= \alpha (\ln \alpha -1)
- \left(\frac{2 b \ln \alpha - 4 b + 2 b \ln b}{\e} \right) \alpha^{3/2} 
+ \left( \frac{c+1}{2} - \frac{b^2}{2 \e^2} \right) \alpha^2 \\
&- \frac{b}{\e} \,\alpha^{5/2} c 
+ \left( \frac{b^2}{2 \e^2}+1 \right) \alpha^3 c + O(\alpha^{7/2} c) + O(\alpha^{5/2}) \, .
\end{align*}
Setting
\[
c = \frac{2 \ln (1/\alpha) + 2 - x \sqrt{\alpha} }{\alpha} 
\]
for constant $x$ causes the terms proportional to $\alpha \ln \alpha$, $\alpha$, and $\alpha^{3/2} \ln \alpha$ to cancel, leaving
\[
\psi(\zeta) 
= \left( \frac{2b (1-\ln b)}{\e} - \frac{x}{2} \right) \alpha^{3/2} + O(\alpha^2) \, .
\]
The coefficient of $\alpha^{3/2}$ is maximized when $b=1$, and is negative whenever $x > 4/\e$.  In that case, $\psi(\zeta_3) < 0$ for sufficiently small $\alpha$, completing the proof.
\end{proofof}

\begin{proofof}{Corollary~\ref{cor:lower}}
First note that 
\begin{equation}
\label{eq:a0}
\alpha_0 = \frac{2}{c} \,W\!\left( \frac{\e c}{2} \right)
\end{equation}
is the root of the equation
\[
c = 2 \,\frac{\ln (1/\alpha_0) + 1}{\alpha_0} \, ,  
\]
since we can also write it as
\[
\e^{c \alpha/2} = \frac{\e}{\alpha_0} \, ,
\]
and multiplying both sides by $c \alpha_0/2$ gives
\[
\frac{c \alpha_0}{2} \,\e^{c \alpha_0/2} = \frac{\e c}{2} \, ,
\]
in which case~\eqref{eq:a0} follows from the definition of $W$.

The root $\alpha$ of 
\[
c = 2 \,\frac{\ln (1/\alpha) + 1}{\alpha} - \frac{x}{\sqrt{\alpha}} 
\]
is then at least 
\[
\frac{2}{c} \,W\!\left( \frac{\e c}{2} \right) + \big(x+o(1)\big) \frac{\partial \alpha_0}{\partial c} \sqrt{\frac{c}{2 \ln c}}
\]
since $\alpha = (1+o(1)) 2 \ln c / c$ and $\partial^2 \alpha_0 / \partial^2 c \ge 0$.  Since
\[
\frac{\partial \alpha_0}{\partial c} = - \big( 1+o(1) \big) \frac{2 \ln c}{c^2} \, , 
\]
the statement follows from Theorem~\ref{thm:lower}.
\end{proofof}

\section*{Acknowledgments}

We are grateful to Amin Coja-Oghlan, Alan Frieze, David Gamarnik, Yuval Peres, Alex Russell, and Joel Spencer for helpful conversations, and to the anonymous reviews for their comments.

%\bibliographystyle{plain}
%\bibliography{randomgraphs}

%% \begin{thebibliography}{99}

%% \bibitem{frieze} A. Frieze, ``On the independence number of random graphs.''  \emph{Discrete Mathematics} 81 (1990) 171--175.

%% \bibitem{coja} A. Coja-Oghlan and C. Efthymiou, ``On independent sets in random graphs.'' \emph{Proc. SODA} (2011).

%% \bibitem{karp-sipser} R. M. Karp and M. Sipser, ``Maximum Matchings in Sparse Random Graphs.'' \emph{Proc. FOCS} (1981) 364--375.

%% \bibitem{ach-peres} D. Achlioptas and Y. Peres, ``The Threshold for Random $k$-SAT is $2k \log 2 - O(k)$ \emph{J. AMS} 17 (2004) 947--973.

%% \bibitem{bollobas} B. Bollob\'as, \emph{Random Graphs}, 2nd edition.  Cambridge, 2001.

%% \bibitem{ach-moore} D. Achlioptas and C. Moore, ``Two moments suffice to cross a sharp threshold.''  \emph{SIAM Journal on Computing} 36 (2006) 740--762.

%\bibitem{friedgut} E. Friedgut, ``Sharp thresholds of graph properties, and the $k$-{SAT} problem.'' \emph{J. AMS} 12 (1999) 1017--1054.

%\end{thebibliography}

\end{document}